\documentclass[11pt]{article}
\usepackage{fullpage}
\usepackage{amssymb}
\usepackage{amsmath,amsthm}
\usepackage{multirow}
\usepackage{color}
\usepackage{xspace}
\usepackage{clrscode}

\usepackage{ifpdf}
\ifpdf    
\usepackage{hyperref}
\else    
\usepackage[hypertex]{hyperref}
\fi
\allowdisplaybreaks

\newcommand{\remove}[1]{}
\providecommand{\card}[1]{\lvert#1\rvert}

\newtheorem{theorem}{Theorem}[section]

\newtheorem{cor}{Corollary}[section]
\newtheorem{lemma}[theorem]{Lemma}

\def\ot{\ensuremath{{\tilde O}}}
\newcommand{\noise}{approximate }

\newcommand{\comment}[1]{}
\newcommand{\junk}[1]{}

\begin{document}
{
\title{Breaking the Variance:  Approximating the Hamming Distance\\ in $\ot(1/\epsilon)$ Time Per Alignment}

\author{
Tsvi Kopelowitz\thanks{University of Michigan. Supported by NSF Grants CCF-1217338, CNS-1318294, and CCF-1514383.}
\and
Ely Porat\thanks{Bar-Ilan University. }
}

\maketitle

\thispagestyle{empty}
\setcounter{page}{0}

\begin{abstract}
The algorithmic tasks of computing the \emph{Hamming distance} between a given pattern of length $m$ and each location in a text of length $n$ is one of the most fundamental algorithmic tasks in string algorithms. Unfortunately, there is evidence that for a text $T$ of size $n$ and a pattern $P$ of size $m$, one cannot compute the exact Hamming distance for all locations in $T$ in time which is less than $\ot(n\sqrt m)$. However, Karloff~\cite{karloff} showed that if one is willing to suffer a $1\pm\epsilon$ approximation, then it is possible to solve the problem with high probability, in $\ot(\frac n {\epsilon^2})$ time.

Due to related lower bounds for computing the Hamming distance of two strings in the one-way communication complexity model, it is strongly believed that obtaining an algorithm for solving the approximation version cannot be done much faster as a function of $\frac 1 \epsilon$. We show here that this belief is \emph{false} by introducing a new $\ot(\frac{n}{\epsilon})$ time algorithm that succeeds with high probability.

The main idea behind our algorithm, which is common in sparse recovery problems, is to reduce the variance of a specific randomized experiment by (approximately) separating heavy hitters from non-heavy hitters. However, while known sparse recovery techniques work very well on vectors, they do not seem to apply here, where we are dealing with
mismatches between {\em pairs} of characters.
We introduce two main algorithmic ingredients. The first is a new sparse recovery method that applies for pair inputs (such as in our setting). The second is a new construction of hash/projection functions, which
allows to count the number of projections that induce mismatches
between two characters {\em exponentially} faster than brute force. We expect that these algorithmic techniques will be of independent interest.

\end{abstract}

\newpage
\section{Introduction}\label{sec:intro}
One of the most fundamental family of problems in string algorithms is to compute the distance between a given pattern $P$ of length $m$ and each location in given larger text $T$ of length $n$ both, over alphabet $\Sigma$, under some string distance metric (See~\cite{LEVENSHTEIN,FP74,ALPU05,LP08,swap-ALP00,AEP04,AABLLPSV06,AHKLP07,PE08,BLPP07,PPZ08,lw-75,AKO10,BJKK04,SP09,PP09,CM02,BI15,CEPR:2012,CP:2010, CEPR:2010,CEPR:2009,AAKLP08,AAILP07,PL07}). The most important distance metric in this setting is the \emph{Hamming Distance} of two strings, which is the number of aligned character mismatches between the strings. Let $\texttt{HAM}(X, Y)$ denote the Hamming distance of two strings $X$ and $Y$.
Abrahamson~\cite{Abrahamson} showed an algorithm whose runtime is $\ot(n\sqrt m)$.
The task of obtaining a faster upper bound seems to be very challenging, and indeed there is a folklore matching conditional lower bound for combinatorial algorithms based on the hardness of combinatorial boolean matrix multiplication (see~\cite{CliffordFolklore}).
However, for constant sized alphabets the runtime can be reduced to $\ot(n)$ using a constant number of convolution computations (which are implemented via the FFT algorithm)~\cite{FP74}.

This naturally lead to approximation algorithms for computing the Hamming distance in this setting, which is the problem that we consider here and is defined as follows.
Denote $T_j = T[j,\ldots,j+m-1]$. In the \textit{pattern-to-text approximate Hamming distance problem}
the input is a parameter $\epsilon >0$, $T$, and $P$. The goal is to compute for all locations $i\in [1, n-m+1]$ a value $\delta_i$ such that $(1-\epsilon)\texttt{\texttt{HAM}}(T_i, P)\leq  \delta_i \leq (1+\epsilon)\texttt{HAM}(T_i, P)$. For simplicity we assume without loss of generality that $\Sigma$ is the set of integers $\{1,2,\ldots,|\Sigma|\}$.

Karloff in~\cite{karloff} utilized the efficiency of the algorithm for constant sized alphabets to introduce a beautiful randomized algorithm for solving the pattern-to-text approximate Hamming distance problem, by utilizing projections of $\Sigma$ to binary alphabets. Karloff's algorithm runs in $\ot(\frac{n}{\epsilon^2})$ time, and is correct with high probability.

\paragraph{Communication complexity lower bounds.}
One of the downsides of Karloff's algorithm is the dependence on $\frac{1}{\epsilon^2}$. In particular, if one is interested in a one percent approximation guarantee, then this term becomes 10000!
However, many believe that beating the runtime of Karloff's algorithm is not possible, mainly since there exist qualitatively related lower bounds for estimating the Hamming distance of two equal length strings (for a single alignment). In particular, Woodruff~\cite{Woodruff04} and later Jayram, Kumar and Sivakumar~\cite{JKS08} showed that obtaining a $(1\pm\epsilon)$ approximation for two strings in the one-way communication complexity model requires sending $\Omega(1/\epsilon^2)$ bits of information. This lower bound implies a lower bound for the sketch size of the Hamming distance and some other streaming problems.

Our results here show that \textbf{this intuition is flawed}, by introducing an $\ot(\frac{n}{\epsilon})$ time algorithm that succeeds with high probability.

\paragraph{The challenge -- beating the variance.}
The main idea of Karloff's algorithm is to project $\Sigma$ to a binary alphabet, and compute the Hamming distance for each location of the projected text and the projected pattern. For a given location denote by $d$ the Hamming distance of the pattern at this location. While the projected Hamming distance is expected to be $\frac d 2$, the variance of the projected  distance could be as high as $\Omega(d^2)$ (see Section~\ref{sec:karloff} for a detailed calculation). In order to overcome this high variance Karloff's algorithm makes use of $\ot(\frac{1}{\epsilon^2})$ projections. More detail is given in Section~\ref{sec:karloff}.

The first step in obtaining a more efficient algorithm is to somehow reduce the number of projections that an algorithm would use. One line of attack would be to somehow reduce the variance. Indeed, such approaches have been considered in other problems~\cite{KNPW11}, and even for the problem considered here, Atallah, Grigorescu, and Wu in~\cite{AGW13} managed to slightly reduce the variance in some cases. They do this by computing the exact contribution to the Hamming distance of the $k$ most frequent characters in the pattern and approximating the contribution of the rest of the characters. This reduces the variance to $O(d\min(m\frac m k,d)$ for each projection, and by repeating the process $k$ times the variance of the average result becomes $d\min(\frac m k,d)/k$. This approach is only useful when the Hamming distance is high -- at least $\frac m k$.

However our goal here is to obtain an even faster algorithm, and so we devise a new method for reducing the variance.

\subsection{Our results and techniques.}
Our main result is the first significant improvement on this problem in the last over 20 years. We present a new randomized algorithm that solves the pattern-to-text approximate Hamming distance with high probability (at least $1-n^{-\Omega(1)}$) that runs in worst-case $\ot(\frac{n}{\epsilon})$ time. This is summarized in the following Theorem.
\begin{theorem}
There exists an algorithm that with high probability solves the pattern-to-text approximate Hamming distance problem and runs in $O(\frac n \epsilon \log \frac 1 \epsilon  \log n \log m \log |\Sigma|)$ time.
\end{theorem}
Furthermore, we introduce two exciting novel techniques in our algorithms.

\paragraph{Intuition.} The main intuition (which may not reveal the technical challenges) is to reduce the variance of the estimation produced by a single projection, by removing pairs of characters with a large contribution to the Hamming distance. In particular, for Hamming distance $d$ we would like to remove pairs that contribute at least $\epsilon d$. Such pairs are called \emph{heavy hitter pairs} and there are at most $\frac 1 \epsilon$ such pairs at each location. Unfortunately, it is not clear how to detect the heavy hitter pairs for each location within the time bounds that we are aiming for.
A common method for quickly approximating the heavy hitters in a given vector is sparse recovery~\cite{GLPS14,GNPRS13,PS12,GLPS10,CCF02,CM04,GSTV07,IR08}. The idea is to approximate the heavy hitters of a vector, while suffering from some extra noise in the form of some additional points in the vector, but not more than $O(\frac 1 \epsilon)$ points.
However, since in our setting the heavy hitters are pairs of elements (as opposed to single elements in a vector), there are structural constraints that make known sparse recovery techniques irrelevant here.

\paragraph{Approximating heavy hitter pairs.} To overcome these structural constraints, we introduce an algorithm that utilizes a \emph{specially constructed small set of projections}. We roughly show that with high probability the $L_2$ distance of the approximate heavy hitter pairs and the actual heavy hitter pairs is small. This in turn allows us to obtain a linear time algorithm which estimates the Hamming distance with variance at most $O(\epsilon d^2)$, and so $\ot(\frac 1 \epsilon)$ repetitions suffice. This construction can be found in Section~\ref{sec:HH}.

\paragraph{Computing all of the projections quickly.} While the approximation of the heavy hitter pairs implies that the number of repetitions can be low, the algorithm that we use will still need to project the $O(\frac 1 \epsilon)$ pairs in the approximation with each of the $\ot(\frac 1 \epsilon)$ projections. Doing this directly will cost $\ot(\frac 1 {\epsilon^2})$ time per location which is too high.
To overcome this, we make use of the way in which our algorithm uses the projected outcomes of the approximating pairs, by constructing a specially tailored set of projections. This construction may be of independent interest, and is detailed in Section~\ref{sec:projections}.
Given $k$ we construct $k$ hash functions $h_1,...h_k$ into the binary alphabet, such that given $x\neq y$ we can count for how many functions $h_i$ we have $h_i(x) = h_i(y)$ in $O(\log k)$ time, rather than $O(k)$ time.

\section{Preliminaries}\label{sec:prelim}
For a given location $j$, consider the alignment of $T_j$ with $P$. This alignment naturally defines an alignment matrix $D=D_j=\{d_{u,v}\}_{u,v\in \Sigma}$, such that for $u\neq v$ we have $d_{u,v}= \card{\{0\leq i \leq m-1 : T[j+i]= u \wedge P[i] = v\}}$ and $0$ otherwise. In words, $d_{u,v}$ is the contribution of the pair $(u,v)$ to the Hamming distance of $T_j$ and $P$.
Clearly, the Hamming distance is $d=\sum_{u,v\in \Sigma} d_{u,v}$.
We emphasize that computing or representing the alignment matrix explicitly is too costly in our setting.

\paragraph{Local versus global operations.}
The operations that our algorithm performs during the computation of the Hamming distance at some location $j$ can be partitioned into two types. The first type are \textit{local} operations which are independent of the computations performed for other locations in $T$. The second type are \textit{global} operations, which in order to be done efficiently may consider the alignments at other locations in $T$. In particular, all of the global operations in our algorithm can be reduced to computing the number of times that a 1 in a projection of $T_j$ to a binary alphabet aligns with a 1 in a projection of $P$ to a binary alphabet (the projection function is not required to be the same for $T_j$ and $P$). Such a computation will make use of the following Theorem.
\begin{theorem}\label{thm:FFT}
Given a binary text $T$ of size $n$ and a binary pattern $P$ of size $m$, there exists an $O(n\log m)$ time algorithm that computes for all locations $i$ in $T$ the number of times that a 1 in $T_i$ is aligned with a $1$ in $P$.
\end{theorem}
The algorithm for Theorem~\ref{thm:FFT} is implemented via a single convolution (using the FFT) in $O(n\log m)$ time, and so can charge an $O(\log m)$ time cost to each location for each global operation. We emphasize that the efficiency of the algorithm in Theorem~\ref{thm:FFT} is only relevant when the projections to binary alphabets are the same for all locations $j$, which indeed will be the case in our algorithm.

\paragraph{Simplifying assumptions.}
With the goal of easing the presentation of our algorithm, we focus on estimating the Hamming distance between $T_j$ and $P$, and show that the number of global and local operations that our algorithm performs for this location is $\ot(1/\epsilon)$. In particular,
we will use $D$ throughout to refer to the alignment matrix $D_j$, omitting the subscript $j$.
We also emphasize that since we are interested in algorithms that succeed with high probability (at least $1-\frac{1}{n^{\Theta(1)}}$) then it suffices to show that with high probability the algorithm succeeds at location $j$.
Furthermore, since we are ignoring poly--log factors, we will only show that the algorithm succeeds with probability which is strictly larger than $\frac 1 2$ by at least some constant. The median of the estimations of $\Theta(\log n)$ independent executions of the algorithm guarantees success with high probability.

\section{Karloff's Algorithm}\label{sec:karloff}

\begin{figure}
\begin{codebox}
\Procname{$\proc{Karloff}(T_j,P,\epsilon)$}
\li $k \leftarrow O(\frac 1 {\epsilon^2})$
\li construct a p.w.i set of 4-wise independent hash functions $h_1,h_2,\ldots, h_{k}$ : $\Sigma\rightarrow \{0,1\}$.
\li \For $i = 1$ to $k$
\li \Do compute $x_i = \texttt{HAM}(h_i(T_j),h_i(P))$.
    \End
\li $X^*=2\frac{\sum_{i=1}^{k}x_i }{k}$
\li return $X^*$
\end{codebox}
\caption{Karloff's Algorithm.}\label{alg:Karloff-Algorithm}
\end{figure}

\begin{figure}
\begin{codebox}
\Procname{$\proc{Approx-Hamm-Distance}(T_j,P,\epsilon)$}
\li $k \leftarrow \frac {8b} {\epsilon}$
\li Construct-$D'(T_j,P,\epsilon)$
\li construct a p.w.i set of 4-wise independent hash functions $h_1,h_2,\ldots, h_{k}$ : $\Sigma\rightarrow \{0,1\}$.
\li $X^*=\sum_{i=1}^{k}\left(\texttt{HAM}(h_i(T_j),h_i(P)) + \frac 1 2 \sum_{\substack{u,v\\ h_i(u)=h_i(v)}}d'_{u,v} - \frac 1 2 \sum_{\substack{u,v\\ h_i(u)\neq h_i(v)}}d'_{u,v}\right) /(k/2)$
\li return $X^*$
\end{codebox}
\caption{The new Algorithm.}\label{alg:new_algo}
\end{figure}

Karloff in \cite{karloff} presented an algorithm for solving pattern-to-text approximate Hamming distance that runs in $\ot(\frac{n}{\epsilon^2})$ time, and is correct with high probability. We present an overview of (a simplified version of) Karloff's algorithm as understanding it is helpful for the setup of the new algorithm presented here. The pseudo-code for the algorithm is given in Figure~\ref{alg:Karloff-Algorithm}.

The only global operation is to compute the Hamming distance in line 4, which happens $O(\frac{1}{\epsilon^2})$ times for a total of $O(\frac{1}{\epsilon^2}\log m)$ time. The rest of the operations are all local and cost $O(\frac{1}{\epsilon^2})$ time.
For any $x_i$ as computed in line 4, the expected value of $x_i$ is $E[x_i]=\frac d 2$, and since each hash function is 4-wise independent, the variance is

{

\begin{align}
\nonumber V[x_i] &= \sum_{\substack{u<v }} V\left[\alpha_{u,v} (d_{u,v}+ d_{v,u})\right]
 = \sum_{\substack{u<v }} E\left[\left(\alpha_{u,v}(d_{u,v}+ d_{v,u}) - E\left[(\alpha_{u,v})(d_{u,v}+ d_{v,u}) \right]\right)^2\right]\\
\nonumber & = \sum_{\substack{u<v }} E\left[\left(\alpha_{u,v}(d_{u,v}+ d_{v,u})  - \frac 1 2 (d_{u,v}+ d_{v,u})\right)^2\right]\\
\nonumber & = \sum_{\substack{u<v }}  \Pr\left[\alpha_{u,v}=1\right]\cdot \left(\frac 1 2 (d_{u,v}+ d_{v,u})\right)^2 +\Pr\left[\alpha_{u,v}=0\right]\cdot \left(  - \frac 1 2 (d_{u,v}+ d_{v,u})\right)^2 \\
\nonumber & = \sum_{\substack{u<v }} \frac 1 4( d_{u,v}+ d_{v,u})^2   = \sum_{\substack{u<v \\ d_{u,v} <\epsilon d \\ d_{v,u} < \epsilon d}}\frac{(d_{u,v}+ d_{v,u})^2}{4} +\sum_{\substack{u<v \\ d_{u,v} \geq \epsilon d \vee d_{v,u} \geq \epsilon d }}\frac{(d_{u,v}+ d_{v,u})^2}{4}\\
& <\frac{\epsilon d^2}{4} + \frac 1 4 \sum_{\substack{u<v \\ d_{u,v} \geq \epsilon d \vee d_{v,u} \geq \epsilon d }}(d_{u,v}+ d_{v,u})^2. \label{eqn:variance_bound}
\end{align}
}

The challenge here is that the variance can be very large -- up to roughly $\Omega(d^2)$, and so in order to overcome it Karloff suggested using $O(\frac{1}{\epsilon^2})$ pair-wise independent projections to the binary alphabet, and then by applying the Chebyshev inequality the probability that the average of the projected distances is an acceptable approximation of $d$ is large enough.

\section{New Algorithm}\label{sec:algo}

The reason why Karloff used $O(\frac{1}{\epsilon^2})$ different projections was because the variance of a single projection is high.
As can be seen in Equation~\ref{eqn:variance_bound}, the high variance is due to the at most $\frac 1 \epsilon$ entries in the alignment matrix $D$ that are larger than $\epsilon d$. Such entries are called \emph{heavy hitters}. If we could somehow separate those heavy hitters from the rest of $D$, and, say, compute their values directly, then the remaining non-heavy hitter entries would have a low variance ($O(\epsilon d^2)$) in which case $O(\frac 1 \epsilon)$ projections would suffice.

However, it is not clear how to find these heavy hitters efficiently. So instead, we use a different strategy.
We say that a matrix $D'=\{d'_{u,v}\}_{u,v\in\Sigma}$, is a \emph{sparse \noise matrix} if:
\begin{enumerate}
\item $\sum_{\substack{u,v\in \Sigma}}(d_{u,v} - d'_{u,v})^2 \leq b\epsilon d^2 $ for constant $b = \frac{2^{13} + 2^{12}  + 1}{{2^{14}}} < 0.752$.
\item The number of non-zero entries in $D'$ is at most $\frac{3}{\epsilon}$.
\end{enumerate}
Notice that the matrix in which all entries that correspond to a heavy hitter pair $(u,v)$ have the value $d_{u,v}$ and all other entries are zero is a sparse \noise matrix, which matches the intuition described above. As we shall show, the only properties of such a matrix that we require are the ones which define a sparse \noise matrix.

Our algorithm will make use of sparse \noise matrices by constructing a (possibly different) matrix for each location in $T$. However, as explained in Section~\ref{sec:prelim} the discussion here focuses on only one location $j$, and so we make use of only one sparse \noise matrix $D'$.
We emphasize that representing $D'$ explicitly is too costly since it is too large. Instead, we use an implicit representation of $D'$ by considering only the non-zero entries.
In Section~\ref{sec:HH} we show an algorithm that with high probability will construct an implicitly represented matrix $D'$ with the desired properties in $\ot(\frac 1 \epsilon)$ time, by proving the following lemma.

\begin{lemma}\label{lem:construct_D'}
There exists an algorithm that with high probability computes a sparse \noise matrix $D'$ for location $j$ in $T$ such that the number of global and local operations performed by the algorithm is $O(\frac 1 \epsilon \log \frac 1 \epsilon \log n \log m\\ \log |\Sigma|)$.
\end{lemma}

Given Lemma~\ref{lem:construct_D'}, the pseudo-code of the algorithm is given in Figure~\ref{alg:new_algo}.
We will now bound the expected value and variance of $X^*$.
Fix $h_i$. Let $\alpha_{u,v} = 1$ if $h_i(u)\neq h_i(v)$ and $0$ otherwise, and let $x_i=\texttt{HAM}(h_i(T_j),h_i(P)) + \frac 1 2 \sum_{h_i(u)=h_i(v)}d'_{u,v} - \frac 1 2 \sum_{h_i(u)\neq h_i(v)}d'_{u,v}$. Notice that $E[\frac 12 - \alpha_{u,v}] = 0$. Recall that each function $h_i$ is 4-wise independent. Then the expected value and variance of $x_i$, under the random choice of $h_i$, is $E[x_i] = \frac d 2$ and
{
\small
\begin{align*}
V[x_i] &= \sum_{\substack{u<v }} V\left[\alpha_{u,v} (d_{u,v}+d_{v,u}) + (\frac 1 2 - \alpha_{u,v})(d'_{u,v}+d'_{v,u})\right]\\
& = \sum_{\substack{u<v }} E\left[\left(\alpha_{u,v}(d_{u,v}+d_{v,u}) + (\frac 1 2 - \alpha_{u,v})(d'_{u,v}+d'_{v,u}) - E\left[(\alpha_{u,v})(d_{u,v}+d_{v,u}) + (\frac 1 2 - \alpha_{u,v})(d'_{u,v}+d'_{v,u})\right]\right)^2\right]\\
& = \sum_{\substack{u<v }} E\left[\left(\alpha_{u,v}(d_{u,v}+d_{v,u}) + (\frac 1 2 - \alpha_{u,v})(d'_{u,v}+d'_{v,u}) - \frac 1 2 (d_{u,v}+d_{v,u})\right)^2\right]\\
& = \sum_{\substack{u<v }}  \left[\Pr\left[\alpha_{u,v}=1\right]\cdot \left(\frac 12 (d_{u,v}+d_{v,u})-\frac 1 2 (d'_{u,v}+d'_{v,u}) \right)^2 +\Pr\left[\alpha_{u,v}=0\right]\cdot \left( \frac 1 2 (d'_{u,v}+d'_{v,u}) - \frac 1 2 (d_{u,v}+d_{v,u})\right)^2\right] \\
& = \sum_{\substack{u<v }} \left[\frac 1 2 \left(\frac 1 2 (d_{u,v}+d_{v,u}) - \frac 1 2 (d'_{u,v}+d'_{v,u})\right)^2+ \frac 1 2 \left(\frac 1 2 (d'_{u,v}+d'_{v,u}) - \frac 1 2 (d_{u,v}+d_{v,u})\right)^2\right]  \\
& = \sum_{\substack{u<v }} \frac 1 4\left((d_{u,v}- d'_{u,v})+(d_{v,u}-d'_{v,u})\right)^2  \leq \sum_{\substack{u<v }} \frac 1 4(2(d_{u,v}- d'_{u,v})^2+2(d_{v,u}-d'_{v,u})^2)  \\
& = \sum_{\substack{u,v }} \frac 1 2(d_{u,v}- d'_{u,v})^2   < \frac{0.752}{2}\epsilon d^2 = 0.376\epsilon d^2.
\end{align*}
}

Thus, given (an implicit) $D'$, we obtain an estimation of $d$ that has low variance.
Recall that the hash functions are pair-wise independent among themselves. Therefore, $E[X^*] = d$ and $V[X^*] =V[\sum_{i=1}^{k} x_i /(k/2)] = \frac{4}{k^2} \sum_{i=1}^{k} V[x_i] =\frac{4V[x_i]}{k} = \frac{4b\epsilon d^2}{8b/\epsilon} = \frac{\epsilon^2 d^2}{2}$.
Therefore, by Chebyshev's inequality, $\Pr[|X^*-E[X^*]| > \epsilon d ] < \frac 1 {2}$.
While $D'$ can be used to decrease the variance of the estimation, there is still a bottleneck in the runtime of the algorithm from computing $\frac 1 2 \sum_{h_i(u)=h_i(v)}d'_{u,v} - \frac 1 2 \sum_{h_i(u)\neq h_i(v)}d'_{u,v}$ for all of the $O(\frac 1 \epsilon)$ projections. This can be done directly in $O(\frac {1}{\epsilon^2})$ time, which is too costly for our goals (since this computation would need to be repeated $O(n)$ times). We overcome this challenge by introducing a special construction of $O(\frac 1 \epsilon)$ projections to a binary alphabet, where each projection is 4-wise-independent, and the projection functions are pair-wise-independent among themselves. The special construction will have the property that computing $$\sum_{i=1}^{k}\left(\frac 1 2 \sum_{h_i(u)=h_i(v)}d'_{u,v} - \frac 1 2 \sum_{h_i(u)\neq h_i(v)}d'_{u,v}\right)$$ can be done in $O(\frac 1 \epsilon)$ time. This construction is summarized in the following lemma, which is proven in Section~\ref{sec:projections}.

\begin{lemma}
There exists an algorithm for constructing $h_1,\ldots,h_k$ in line 3 of Approx-Hamming-Distance, so that executing line 4 of Approx-Hamming-Distance takes  $O(\frac 1\epsilon \log k)$ time.
\end{lemma}

Thus, the total runtime of the algorithm is $\ot(\frac{1}{\epsilon})$ time per location, for a total of $\ot(\frac{n}{\epsilon})$ time for all locations.

\section{Computing $D'$}\label{sec:HH}
\paragraph{Intuition.} In order to construct a sparse \noise matrix $D'$ we would intuitively like to estimate the at most $\frac{1} \epsilon$ heavy hitters of $D$ which are the entries of $D$ that are at least $\epsilon d$. One idea for obtaining this estimation is to project the alphabet $\Sigma$ to a smaller alphabet, with the hopes that heavy hitters before the projection can be established from the heavy hitters after the projection. However, this specific task seems out of grasp, since the projections introduce too much noise. To overcome this challenge, we use several specially constructed projections so that together with the proper algorithm we are able to estimate $O(\frac{1}{\epsilon})$ entries of $D$, which will suffice in order to bound $\sum_{\substack{u,v\in \Sigma}} (d_{u,v} - d'_{u,v})^2$. We emphasize that while our algorithm may in fact not estimate all of the heavy hitters, this bound is still obtained with high probability.

\begin{figure}
\begin{codebox}
\Procname{$\proc{Construct-}D'(T_j,P,\epsilon)$}
\li implicitly initialize all $d'_{u,v} \leftarrow \infty$
\li \For $i=0 $ to $\log \frac {1}{\epsilon}$
\li \Do $\ell_i \leftarrow 32\cdot 2^i$
\li     $r_i \leftarrow \frac{32}{2^i\epsilon}$
\li     \Repeat $\Theta(\log n)$ times
\li         \If $\ell_i \geq r_i$
\li         \Then   pick a random 4-wise independent projection $\tau_i:\Sigma \rightarrow [\ell_i]$
\li                 $\pi_i(\cdot) = \tau_i(\cdot ) \mod r_i$

\li         \Else   pick a random 4-wise independent projection $\pi_i:\Sigma \rightarrow [r_i]$
\li                 $\tau_i(\cdot) = \pi_i(\cdot ) \mod \ell_i$
            \End
\li         \For every pair $(x,y)\in [\ell_i] \times [r_i]$ where no $(\sigma,\sigma)$ is in the preimage of $(x,y)$
\li         \Do $c_{x,y} \leftarrow$ number of times that $x$ in $\tau_i(T_j)$ is aligned with $y$ in $\pi_i(P)$
\li             \If there exists $(u,v)$ in preimage of $(x,y)$ where $d_{u,v}> \frac{c_{x,y}}{2}$
\li             \Then $d'_{u,v} \leftarrow \min(d'_{u,v},c_{x,y})$
                \End
            \End
        \End
    \End
\li implicitly set all $d'_{u,v} = \infty$ to $d'_{u,v} \leftarrow 0$
\li \For every $d'_{u,v}> 0$ such that $d'_{u,v}$ is not one of the $\frac 3 \epsilon$ largest values in $D'$
\li \Do $d'_{u,v}\leftarrow 0$
    \End
\end{codebox}
\caption{Constructing D'}\label{alg:construct_D'}
\end{figure}

\paragraph{The projections.} We consider $O(\log \frac 1 \epsilon)$ 4-wise independent projections as follows\footnote{Notice that these projections are the same for all locations in $T$ and not just for location $j$. This enables the use of Theorem~\ref{thm:FFT} for global operations in our setting.}. The $i$th projection, for $i=0,\ldots, \log \frac 1 \epsilon$, is defined by two projection functions $\tau_i: \Sigma \rightarrow [\ell_i]$ and $\pi_i: \Sigma \rightarrow [r_i]$, where $\ell_i = 32\cdot 2^i$ and $r_i = \frac{32}{2^i\epsilon}$. Notice that for all $i$ we have $\ell_i\cdot r_i = \frac{32^2}{\epsilon}$. We assume without loss of generality that $\frac{32^2}{\epsilon}$ is a power of 2.
For each $i$, the text $T_j$ is projected with $\tau_i$ while the pattern $P$ is projected with $\pi_i$.
We then compute for each pair $(x,y)\in [\ell_i]\times [r_i]$ the number of times that $x$ in the projected text is aligned with $y$ in the projected pattern using a single global convolution. Since the number of such pairs is $O(\frac 1 \epsilon)$ this can be computed in $O(\frac{\log m}{\epsilon})$ time for each $i$ via global operations, for a total of $O(\frac{\log m}{\epsilon}\log \frac 1 \epsilon)$ for all $i$.

Intuitively, we would like to claim that for a pair of characters $u\neq v$ the number of times that $\tau_i(u)$ in the projected text aligns with $\pi_i(v)$ in the projected pattern is a close enough estimation of $d_{u,v}$.
However, if we allow $\tau_i$ and $\pi_i$ to be any random projection functions then we will be introducing new mismatches via the projection, since it is likely that for a projected pair $(x,y)$ where $x\neq y$ there exists some character $\sigma\in\Sigma$ such that $(x,y)= (\pi_i(\sigma) ,\tau_i(\sigma))$, and the task of distinguishing projections of matching pairs from projections of mismatching pairs seems to be too difficult within the allowed time.
One method for overcoming this problem would be to only consider the number of mismatches for projected pairs as long as they are not of the form $(\pi_i(\sigma), \tau_i(\sigma))$ for any $\sigma\in\Sigma$. Then we would hope that repeating the entire process with enough choices of $\pi_i$ and $\tau_i$ will guarantee that all of the appropriate pairs of different characters are projected enough times to a pair that is not of this form. However, for each $\sigma\in \Sigma$,  $(\pi_i(\sigma),\tau_i(\sigma)$ is a uniformly random point in $[l_i]\times [r_i]$, which is a universe of size $O(\frac 1 \epsilon)$. Given $O(\frac 1 \epsilon \log { \frac 1 \epsilon})$ random points in this universe with high probability each point appears at least once (this is the coupon collector problem). If $|\Sigma|=\Omega(\frac 1 \epsilon \log \frac 1 \epsilon)$ then with high probability we will fail to recover any pair in the preimage.

Instead, we do the following. Assume that $\ell_i \geq r_i$. The other case of $\ell_i <r_i$ is dealt with by reversing the roles of $\pi_i$ and $\tau_i$. We first pick a random $\tau_i$, and then
set $\pi_i(\sigma)=\tau_i(\sigma) \mod r_i$. Recall that both $\tau_i$ and $r_i$ are powers of $2$. The following lemma bounds the probability that a pair the projection of a given pair of different characters $(u,v)$ is discarded.
\begin{lemma}
If $r_i \leq \ell _i$, then for a given pair of different characters $(u,v)$ the probability that this pair is projected to a pair of the form $(\tau_i(\sigma), \pi_i(\sigma))$ for some $\sigma\in\Sigma$ is at most $ \frac 1 {32}$.
\end{lemma}
\begin{proof}

{
\begin{align*}
\Pr[\exists_{\sigma \in \Sigma} : \tau_i(u)=\tau_i(\sigma) \wedge \pi_i(v)=\pi_i(\sigma)]  & \leq \Pr[\exists_{\sigma \in \Sigma} : \pi_i(u)=\pi_i(\sigma) \wedge \pi_i(v)=\pi_i(\sigma)]\\
& \leq \Pr[\pi_i(u)=\pi_i(v)] = \frac 1 {r_i} \leq \frac 1 {32}
\end{align*}
}

\end{proof}
Therefore, repeating the entire process with $O(\log n)$ choices of $\pi_i$ and $\tau_i$ will guarantee that with high probability every pair of different characters is projected $\Omega(\log n)$ times to a pair that is not of this form (we will need this $\Omega(\log n)$ repetition later).

\paragraph{The Algorithm.}
The pseudo-code for the algorithm is given in Figure~\ref{alg:construct_D'}. In lines 6-10 we create the projection functions, and then in lines 11-12 we compute $c_{x,y}$ for each projected pair $(x,y)$ that is not of the form
$(\pi_i(\sigma), \tau_i(\sigma))$ for any $\sigma\in\Sigma$,
where $c_{x,y}$ is the exact number of alignments of this projected pair.
We then use a bit-tester scheme using error-correcting codes (see~\cite{GLPS10}) in line 13 in order to establish if there is a pair $(u,v)$ in the preimage of $(x,y)$ that contributed more than half of $c_{x,y}$. If so, then in line 14 we use $c_{x,y}$ to estimate $d_{u,v}$, unless a smaller estimator was encountered before. Finally, the algorithm filters away all but the largest $\frac 3 \epsilon$ entries in $D'$, thereby guaranteeing that $D'$ is sparse enough.

\paragraph{Projected-noise.}
Consider a pair $(u,v)$. For a given $\pi_i$ and $\tau_i$ let $$\sum_{\substack{(u',v') \neq (u,v): \tau_i(u)=\tau_i(u') \wedge \pi_i(v)=\pi_i(v')}}d_{u',v'}$$ be the \textit{projected-noise} for $(u,v)$.
We would like to bound the amount of projected-noise for a pair $(u,v)$, since if it is less than $d_{u,v}$ then the \textit{bit-tester} in line 13 will identify the pair $(u,v)$ from the projected noise.
In the following analysis we focus on a specific choice of $i$, $\tau_i$, $\pi_i$, and a pair $(u,v)$. The analysis is partitioned into two cases. In the first case we consider pairs $(u',v')$ of the form either $(u,v')$ or $(u',v)$.
In the second case we consider the remaining possible forms.

\paragraph{The first case.}
We focus on pairs of the form either $(u,v')$ or $(u',v)$. Let $w(u) = \sum_{v'\in \Sigma} d_{u,v'}$, and let $w(v) = \sum_{u'\in \Sigma} d_{u',v}$.
Notice that for any pair $(u,v')$ the probability that $\pi_i(v)=\pi_i(v')$ is $\frac 1 {r_i}$. Similarly, for any pair $(u',v)$ the probability that $\tau_i(u)=\tau_i(u')$ is $\frac 1 {\ell_i}$.
Therefore,
$E[\sum_{\substack{v' \neq v\\ \pi_i(v')=\pi_i(v)}}d_{u,v'}] = \frac{w(u)}{r_i}$ and $E[\sum_{\substack{u' \neq u\\ \pi_i(u')=\pi_i(u)}}d_{u',v}] = \frac{w(v)}{\ell_i}$. Thus the expected amount of projected noise from pairs of the form $(u,v')$ or $(u',v)$ is $\frac{w(u)}{r_i} + \frac {w(v)}{\ell_i}$.
Recall that $\ell_i\cdot r_i = \frac{32^2}{\epsilon}$. This expected projected noise is minimized when $\ell_i=\sqrt{\frac{\epsilon w(u)}{32^2 w(v)}}$, and will then be $\frac{w(u)}{r_i} + \frac {w(v)}{\ell_i} = \frac{\sqrt{\epsilon w(u) w(v)}}{16}$. However, our algorithm is restricted to values of $\ell_i$ and $r_i$ that are powers of 2. Within these limited options for $\ell_i$ and $r_i$ the new minimized expected amount of projected noise is at most twice the minimum over all options for $\ell_i$ and $r_i$, and hence the expected amount of noise due to this type of pair is at most $\frac{\sqrt{\epsilon w(u) w(v)}}{8}$.

\paragraph{The second case.}
We now focus on the remaining cases. Recall that the projections are 4-wise independent. Furthermore, recall that our algorithm ignores projected pairs of the form $(\tau_i(\sigma), \pi_i(\sigma))$ for any $\sigma\in\Sigma$. Therefore, we can ignore the case $u=v'$ since otherwise if $\tau_i(u)= \tau_i(u')$ and $\pi_i(v) = \pi_i(v')$ then $(\tau_i(u),\pi_i(v))= (\tau_i(v'),\pi_i(v'))$ and the algorithm skips this projected pair. Similarly, we can ignore the case $u'=v$. Therefore, the only case that remains is that all four of $u,u',v,v'$ are distinct. Since the projections are 4-wise independent, it must be that
{
\small
\begin{align*}
\Pr[\tau_i(u)= \tau_i(u') \wedge \pi_i(v)= \pi_i(v')]  =  \frac{1}{\ell_i} \cdot \frac{1}{r_i} = \frac{\epsilon}{2^{10}}.
\end{align*}
Thus the expected amount of such noise due to this type of pair is at most $\frac{\epsilon d}{2^{10}}$.
}
\paragraph{Bounding the noise.}
Now, via Markov's inequality, with probability at least $\frac 1 2$ the amount of noise on $(u,v)$ is at most $\frac{\sqrt{\epsilon w(u) w(v)}}{4}+ \frac{\epsilon d}{2^9}$. Since each pair $(u,v)$ is not projected into a pair of the form $\pi_i(\sigma), \tau_i(\sigma))$ for any $\sigma\in \Sigma$ at least $\Omega(\log n)$ times, and we pick the smallest estimation out of all choices of projections, then with high probability the total amount of noise on $(u,v)$ is at most $\frac{\sqrt{\epsilon w(u) w(v)}}{4}+ \frac{\epsilon d}{2^9}$.

The amount of noise on $(u,v)$ can have two types of estimations for values of $d_{u,v}$. The first type of estimation comes from the case in which the noise on $(u,v)$ is smaller than $d_{u,v}$. In this case the bit-tester will find $(u,v)$ and so $d'_{u,v}\leq d_{u,v} + \frac{\sqrt{\epsilon w(u) w(v)}}{4}+ \frac{\epsilon d}{2^9}$. In other words, this estimation overestimates the contribution of $d_{u,v}$ by at most the amount of noise on $(u,v)$.
The second type of estimation comes from the case in which the noise on $(u,v)$ is at least $d_{u,v}$ and so $d'_{u,v}=0$. This case will only concern us if $(u,v)$ is a heavy hitter pair (in which case we may give this pair an estimation of $0$). This case implies that the bit-tester failed to find $(u,v)$, and so $\frac{\sqrt{\epsilon w(u) w(v)}}{4}+ \frac{\epsilon d}{2^9} \geq d_{u,v}$, or  $d'_{u,v} = 0 \geq d_{u,v}-\frac{\sqrt{\epsilon w(u) w(v)}}{4}+ \frac{\epsilon d}{2^9}  $.

\paragraph{Filtering.}
The last part of the algorithm for constructing $D'$ filters all but the largest $\frac 3 \epsilon$ estimators. We will now show that this filtering does not affect the estimation of any heavy hitter pair. If the estimator in $D'$ for a heavy hitter pair was 0 prior to the filtering, then the estimator after the filtering is clearly unchanged. What needs to be proven is that if the estimator in $D'$ for a heavy hitter was larger than 0 prior to the filtering, then with high probability that estimator will be one of the $\frac 3 \epsilon$ largest entries, and so it remains unchanged due to the filtering. For the rest of the discussion here we assume that the values in $D'$ are the values prior to the filtering.

By a simple counting argument, there are at most $\frac 2 \epsilon$ entries in $D$ with a value that is at least $\frac \epsilon 2 d$. The following lemma will help us complete the proof.

\begin{lemma}
With high probability, there are at most $\frac 1 \epsilon$ entries in $D$ with value less than $\frac \epsilon 2$ that are estimated in $D'$ with value at least $\epsilon d$.
\end{lemma}
\begin{proof}
Recall that there are two cases that contribute noise. With high probability the first case contributes $\frac{\sqrt{\epsilon w(u) w(v)}}{4}$ and the second case contributes $\frac {\epsilon d} {2^9}$. Let $S$ be the set of pairs for which the noise from the first case is more than $\frac{ \epsilon d} 4$, which is required (but not sufficient) in order for their estimation to be at least $\epsilon d$.
Then
{
\begin{align*}
|S|\left(\frac{ \epsilon d} 4\right)^2 & \leq \sum_{(u,v)\in S} \left(\frac{\sqrt{\epsilon w(u) w(v)}}{4} \right)^2
 = \sum_{(u,v) \in S} \frac {\epsilon w(u) w(v)}{{16}}
 = \frac{\epsilon} {{16}}   \sum_{u,v\in \Sigma} {w(u) w(v)}
 = \frac{\epsilon d^2} {{16}}.
\end{align*}
}
Therefore, $|S| \leq \frac 1 \epsilon$.
\end{proof}

\begin{cor}
Prior to filtering there are at most $\frac 3 \epsilon$ entries in $D'$ with value at least $\epsilon d$ is at most.
\end{cor}

Finally, notice that a heavy hitter pair $(u,v)$ such that $d'_{u,v} > 0$ must have $d'_{u,v} \geq \epsilon d$, and so this entry cannot be filtered away.

\paragraph{Bounding the variance.}
Notice that there are at most $\frac 3 \epsilon$ non-zero entries in $D'$ and at most $\frac 1 \epsilon$ heavy hitters, and recall that some heavy hitter pairs may be estimated in $D'$ with 0. Finally, with high probability we have

{
\small
\begin{align*}
\sum_{\substack{u,v\in \Sigma}} (d_{u,v}-d'_{u,v})^2 & = \sum_{\substack{u,v\in \Sigma \\ d_{u,v} < \epsilon d \\ d'_{u,v}=0}} (d_{u,v})^2 + \sum_{\substack{u,v\in \Sigma \\ d_{u,v}\geq \epsilon d\vee d'_{u,v}\neq 0} } (d_{u,v}-d'_{u,v})^2\\
&\leq  \frac{\epsilon d^2}{2} + \sum_{\substack{u,v\in \Sigma \\ d_{u,v}\geq \epsilon d\vee d'_{u,v}\neq 0}}\left(\frac{\sqrt{\epsilon \cdot w(u) \cdot w(v)}}{4}+ \frac{\epsilon d }{{2^9}}\right)^2\\
&\leq \frac{\epsilon d^2}{2} + \sum_{\substack{u,v\in \Sigma \\ d_{u,v}\geq \epsilon d\vee d'_{u,v}\neq 0}}\left(\frac{\sqrt{\epsilon \cdot w(u) \cdot w(v)}}{2}\right)^2+ \sum_{\substack{u,v\in \Sigma \\ d_{u,v}\geq \epsilon d\vee d'_{u,v}\neq 0}}\left(\frac{\epsilon d}{{2^8}}\right)^2\\
& \leq \frac{\epsilon d^2}{2} + \left(\sum_{\substack{u,v\in \Sigma \\ d_{u,v}\geq \epsilon d\vee d'_{u,v}\neq 0}}\frac{\epsilon \cdot w(u) \cdot w(v)}{{4}}\right)+ \frac{4\epsilon d^2}{{2^{16}}}\\
& \leq \frac{\epsilon d^2}{2} + \frac{\epsilon d^2}{{4}}+ \frac{\epsilon d^2}{{2^{14}}}
 \leq \frac{2^{13} + 2^{12}  + 1}{{2^{14}}}\epsilon d^2
 = b\epsilon d^2
\end{align*}
}

\section{Constructing the Projections}\label{sec:projections}
We now turn our focus towards constructing the projections in line 4 from Figure~\ref{alg:new_algo}, so that the computation in line 5 will take $\ot(\frac{1}{\epsilon})$ time. In particular, we will show that computing $$\sum_{i=1}^{k}\left(\frac 1 2 \sum_{\substack{u,v \\ h_i(u)=h_i(v)}}d'_{u,v} - \frac 1 2 \sum_{\substack{u,v \\ h_i(u)\neq h_i(v)}}d'_{u,v}\right)$$ can be done in $\ot(\frac{1}{\epsilon}\log k)$ time.

To start off, notice that it suffices to know for each $d'_{u,v} > 0$ the number $\beta_{u,v} = |\{1\leq i \leq k : h_i(u) = h_i(v) \}|$, since
{
\small
\begin{align*}
\sum_{i=1}^{k}\left(\sum_{\substack{u,v \\ h_i(u)=h_i(v)}}d'_{u,v} - \sum_{\substack{u,v \\ h_i(u)\neq h_i(v)}}d'_{u,v}\right) =\sum_{d'_{u,v} > 0} (\beta_{u,v} - (k-\beta_{u,v})) d'_{u,v}.
\end{align*}
}
Once we show how to construct the hash functions $h_1,\ldots, h_k$ so that we can compute $\beta_{u,v}$ for $d'_{u,v} \neq 0 $ in $O(\log k)$ time, then we are done.

\paragraph{The construction.} Assume without loss of generality that $k$ is a power of 2. Consider $2\log k$  base hash functions $f_1,\ldots f_{2\log k}: \Sigma \rightarrow [0,1]$ where each base hash function is picked independently from a 4-wise independent family of hash functions. We partition the base hash functions to $\log k$ pairs, and consider all possible combinations of picking 1 function from each pair. Each combination of $\log k$ functions defines a different projection hash function $h_i$ by considering the xor of the outputs of the $\log k$ functions. The number of projection hash functions is exactly the number of combinations of base hash functions in our construction, which is $(2)^{\log k} = k$, as required.

We now argue that each projection hash function $h_i$ is 4-wise independent, and that the projection hash functions are pair-wise independent among themselves.
Since $h_i$ is the xor of $\log k$ 4-wise independent base functions, $h_i$ is 4-wise independent as well.
Moreover, for each $h_i\neq h_j$, there must be at least one pair of base functions in which $h_i$ uses one base function and $h_j$ uses the other base function. Since these base functions are independent, $h_i$ and $h_j$ must be independent as well.

\paragraph{Computing $\beta_{u,v}$.} Set $(u,v)$. Build a balanced binary tree over the $\log k$ pairs of base functions and compute $f_1(u), f_2(u),\ldots f_{2\log k} (u)$ and $f_1(v), f_2(v),\ldots f_{2\log k} (v)$. For each node $w$ in the balanced binary tree let $t_w$ be the number of pairs of base functions at the leaves of the subtree of $w$. For each such $t_w$ pairs of base functions, we consider all $2^{t_w}$ combinations of picking one base function from each pair. Each such combination defines a projection function local to $w$ by taking the xor of the outputs of the base hash functions for that combination.
Let $e_w$ be the number of such local projections for which the projection on $u$ and the projection on $v$ are equal, while $d_w$ is the number of such projections for which the projection on $u$ and the projection on $v$ are not equal.
If $\ell$ and $r$ are the left and right children of $w$, respectively, then $e_w = e_\ell \cdot e_r + d_\ell \cdot d_r$ and $d_w = e_\ell\cdot d_r + d_\ell \cdot e_r$. Finally, since the root of the balanced binary tree $root$ covers all of the $\log k$ pairs of base functions, then $\beta_{u,v} = e_{root}$. Thus using a direct bottom up approach we can compute $\beta_{u,v}$ in $O(\log k)$ time.

\bibliographystyle{plain}
\bibliography{erc}

\end{document}